\documentclass[journal,draftclsnofoot,onecolumn]{IEEEtran}

\usepackage{setspace}
\usepackage{amsmath}
\usepackage{amssymb}
\usepackage{mathrsfs}
\usepackage{amsthm}
\usepackage{multirow}
\usepackage{color}
\usepackage{longtable}
\usepackage{array}
\usepackage{url}
\usepackage{comment}
\usepackage{enumerate}
\usepackage{eucal}
\usepackage{graphics}

\usepackage{algorithm}
\usepackage{algpseudocode}
\usepackage{cite}
\usepackage[table]{xcolor}

\usepackage{stmaryrd}

\usepackage{mathtools}
\usepackage{xparse}
\DeclarePairedDelimiterX{\set}[1]{\{}{\}}{\setargs{#1}}
\NewDocumentCommand{\setargs}{>{\SplitArgument{1}{;}}m}
{\setargsaux#1}
\NewDocumentCommand{\setargsaux}{mm}
{\IfNoValueTF{#2}{#1} {#1\,\delimsize|\,\mathopen{}#2}}

\DeclarePairedDelimiter\abs{\lvert}{\rvert}
\DeclarePairedDelimiter\ceil{\lceil}{\rceil}
\DeclarePairedDelimiter\floor{\lfloor}{\rfloor}
\DeclarePairedDelimiter\parenv{\lparen}{\rparen}

\newcommand{\cB}{\mathcal{B}}

\newcommand{\cG}{\mathcal{G}}

\newcommand{\cS}{\mathcal{S}}

\newcommand{\bE}{\mathbb{E}}

\renewcommand{\le}{\leqslant}
\renewcommand{\leq}{\leqslant}
\renewcommand{\ge}{\geqslant}
\renewcommand{\geq}{\geqslant}

\theoremstyle{plain}
\newtheorem{theorem}{Theorem}
\newtheorem{corollary}[theorem]{Corollary}

\theoremstyle{definition}
\newtheorem{definition}[theorem]{Definition}

\newtheorem*{remark}{Remark}

\newcommand{\F}{\mathbb{F}}
\newcommand{\R}{\mathbb{R}}
\newcommand{\Z}{\mathbb{Z}}
\newcommand{\N}{\mathbb{N}}

\newcommand{\ve}{\mathbf{e}}

\newcommand{\vv}{\mathbf{v}}

\newcommand{\vx}{\mathbf{x}}
\newcommand{\vy}{\mathbf{y}}
\newcommand{\vz}{\mathbf{z}}

\newcommand{\Zero}{{\mathbf{0}}}

\newcommand{\opsi}{\overline{\psi}}
\DeclareMathOperator{\wt}{wt}
\DeclareMathOperator{\vol}{vol}

\newcommand{\kp}{k_+}
\newcommand{\km}{k_-}
\newcommand{\BALL}{{\mathcal B}(n,t,\kp,\km)}
\newcommand{\eqdef}{\triangleq}
\newcommand{\splt}{\diamond}

\title{On Lattice Packings and Coverings of Asymmetric
  Limited-Magnitude Balls}

\author{
  Hengjia Wei, Xin Wang, and Moshe Schwartz~\IEEEmembership{Senior Member,~IEEE}%
  \thanks{Hengjia Wei is with the School
    of Electrical and Computer Engineering, Ben-Gurion University of the Negev,
    Beer Sheva 8410501, Israel
    (e-mail: hjwei05@gmail.com).}%
  \thanks{Xin Wang is with the Department of Mathematics,
    Soochow University, Suzhou 215006, China
    (e-mail: xinw@suda.edu.cn).}
  \thanks{Moshe Schwartz is with the School
    of Electrical and Computer Engineering, Ben-Gurion University of the Negev,
    Beer Sheva 8410501, Israel
    (e-mail: schwartz@ee.bgu.ac.il).}%
  \thanks{H. Wei and M. Schwartz were supported in part by an Israel
    Science Foundation (ISF) grant 270/18.  The research of X. Wang
    was supported by the National Natural Science Foundation of China
    under Grant No. 11801392 and the Natural Science Foundation of
    Jiangsu Province under Grant No. BK20180833.  }
} 
 
\begin{document}

\maketitle

\begin{abstract}
  We construct integer error-correcting codes and covering codes for
  the limited-magnitude error channel with more than one error. The
  codes are lattices that pack or cover the space with the appropriate
  error ball. Some of the constructions attain an asymptotic
  packing/covering density that is constant. The results are obtained
  via various methods, including the use of codes in the Hamming
  metric, modular $B_t$-sequences, $2$-fold Sidon sets, and sets
  avoiding arithmetic progression.
\end{abstract}


\section{Introduction}

\IEEEPARstart{S}{everal} applications use information that is encoded
as vectors of integers, either directly or indirectly. Furthermore,
these vectors are affected by noise that may increase or decrease
entries of the vectors by a limited amount. We mention a few of these
examples: In high-density magnetic recording channels, information is
stored in the lengths of runs of $0$'s. Various phenomena may cause
the reading process to shift the positions of $1$'s (peak-shift
error), thereby changing the length of adjacent runs of $0$'s by a
limited amount (e.g., see \cite{KuzVin93,LevVin93}). In flash
memories, information is stored in the charge levels of cells in an
array. However, retention (slow charge leakage), and inter-cell
interference, may cause charge levels to move, usually, by a limited
amount (e.g., see~\cite{CasSchBohBru10}). More recently, in some
DNA-storage applications, information is stored in the lengths of
homopolymer runs.  These however, may end up shorter or longer than
planned, usually by a limited amount, due to variability in the
molecule-synthesis process (see~\cite{JaiFarSchBru20}).

In all of the applications mentioned above, an integer vector
$\vv\in\Z^n$ encodes information. If at most $t$ of its entries suffer
an increase by as much as $\kp$, or a decrease by as much as $\km$, we
can write the corrupted vector as $\vv+\ve$, where $\ve$ resides
within a shape we call the \emph{$(n,t,\kp,\km)$-error-ball}, and is
defined as
\begin{equation}
  \label{eq:ball}
  \BALL\triangleq \set*{\vx=(x_1,x_2,\ldots,x_n)\in \Z^n ; -\km \leq x_i \leq \kp \text{ and  } \wt(\vx) \leq t  },
\end{equation}
where $\wt(\vx)$ denotes the Hamming weight of $\vx$.

It now follows that an error-correcting code in this setting is
equivalent to a packing of $\Z^n$ by $\BALL$, a covering code is
equivalent to a covering of $\Z^n$ by $\BALL$, and a perfect code is
equivalent to a tiling of $\Z^n$ by $\BALL$.  An example of
$\cB(3,2,2,1)$ is shown in Fig.~\ref{fig:qc}.

\begin{figure}
  \begin{center}
    \includegraphics{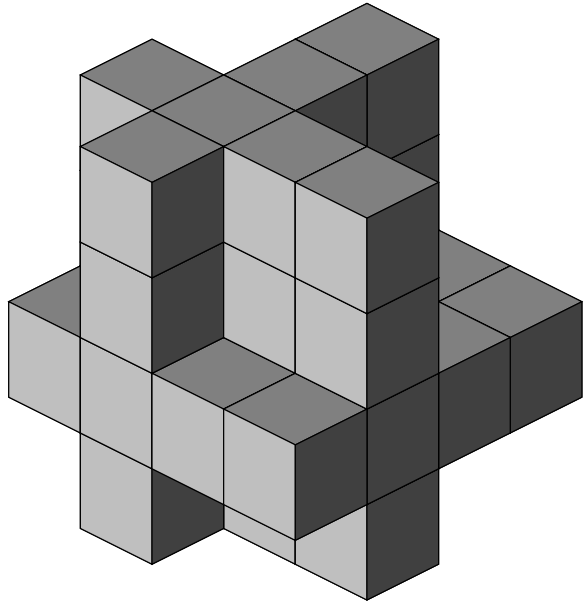}
  \end{center}
  \caption{A depiction of $\cB(3,2,2,1)$ where each point in
    $\cB(3,2,2,1)$ is shown as a unit cube.}
  \label{fig:qc}
\end{figure}

A significant amount of works has been devoted to
tiling/packing/covering of $\Z^n$ by these shapes, albeit, almost
exclusively for the case of $t=1$. When packing and tiling are
concerned, the \emph{cross}, $\cB(n,1,k,k)$, and semi-cross,
$\cB(n,1,k,0)$ have been extensively researched, e.g., see
\cite{Ste84,HamSte84,HicSte86,SteSza94,KloLuoNayYar11} and the many
references therein. This was extended to \emph{quasi-crosses},
$\cB(n,1,\kp,\km)$, in \cite{Sch12}, creating a flurry of activity on
the subject
\cite{YarKloBos13,Sch14,ZhaGe16,ZhaZhaGe17,ZhaGe18,YeZhaZhaGe20}. To
the best of our knowledge,
\cite{Ste90,KloLuoNayYar11,BuzEtz12,WeiSchwartz} are the only works to
consider $t\geq 2$. \cite{Ste90,BuzEtz12} considered tiling a notched
cube (or a ``chair''), which for certain parameters becomes
$\cB(n,n-1,k,0)$, while \cite{KloLuoNayYar11} considered packing the
same ball $\cB(n,n-1,k,0)$.  Among others, \cite{WeiSchwartz} recently
studied the tiling problem in the most general case, i.e., tiling
$\BALL$ for $t \geq 2$. Covering problems have also been studied,
though only when $t=1$,~\cite{CSW2014,Klove2016,KloveSchwartz2014}.

The main goal of this paper is to study packing and covering of $\Z^n$
by $\BALL$ when $t\geq 2$. We provide explicit constructions for both
packings and coverings, as well some non-constructive existence
results.  In particular, we demonstrate the existence of packings with
asymptotic packing density $\Omega(1)$ (as $n$ tends to infinity) for
some sets of $(t,\kp,\km)$, and the existence of coverings with
density $O(1)$ for any given $(t,\kp,\km)$. Additionally, we
generalize the concept of packing to $\lambda$-packing, which works in
conjunction with the list-decoding framework and list size $\lambda$.
We show the existence of $\lambda$-packings with density
$O(n^{-\epsilon})$ for any $(t,\kp,\km)$ and arbitrarily small
$\epsilon>0$, while maintaining a list size
$\lambda=O(\epsilon^{-t})$, which does not depend on $n$. Our results
are summarized at the end of this paper, in Table~\ref{tab:summary}.

The paper is organized as follows. We begin, in
Section~\ref{sec:prelim}, by providing notation and basic known
results used throughout the paper. Section~\ref{sec:packing} is
devoted to the study of packings. This is generalized in
Section~\ref{sec:genpack} to $\lambda$-packings. In
Section~\ref{sec:covering} we construct coverings. Finally, we
conclude in Section~\ref{sec:conclusion} by giving a summary of the
results as well as some open problems.

\section{Preliminaries}
\label{sec:prelim}

For integers $a\leq b$ we define $[a,b]\eqdef\set*{a,a+1,\dots,b}$ and
$[a,b]^*\eqdef [a,b]\setminus\set*{0}$. We use $\Z_m$ to denote the
cyclic group of integers with addition modulo $m$, and $\F_q$ to
denote the finite field of size $q$.

A \emph{lattice} $\Lambda\subseteq\Z^n$ is an additive subgroup of
$\Z^n$ (sometimes called an \emph{integer lattice}). A lattice
$\Lambda$ may be represented by a matrix $\cG(\Lambda)\in\Z^{n\times
  n}$, the span of whose rows (with integer coefficients) is
$\Lambda$. From a geometric point of view, when viewing $\Lambda$ inside
$\R^n$, a \emph{fundamental region} of $\Lambda$ is defined as
\[ \Pi(\Lambda)\eqdef \set*{ \sum_{i=1}^n c_i\vv_i ; c_i\in\R, 0\leq c_i < 1 },\]
where $\vv_i$ is the $i$-th row of $\cG(\Lambda)$. It is well known
that the volume of $\Pi(\Lambda)$ is $\abs{\det(\cG(\Lambda))}$, and
is independent of the choice of $\cG(\Lambda)$. We therefore denote
\[ \vol(\Lambda)\eqdef \vol(\Pi(\Lambda)) = \abs*{\det(\cG(\Lambda))}.\]
In addition, if $\vol(\Lambda)\neq 0$ then
\[ \vol(\Lambda)=\abs*{\Z^n / \Lambda}.\]

We say $\cB\subseteq\Z^n$ \emph{packs} $\Z^n$ by $T\subseteq\Z^n$, if
the translates of $\cB$ by elements from $T$ do not intersect,
namely, for all $\vv,\vv'\in T$, $\vv\neq\vv'$,
\[ (\vv+\cB)\cap(\vv'+\cB)=\varnothing.\]
We say $\cB$ \emph{covers} $\Z^n$ by $T$ if
\[ \bigcup_{\vv\in T} (\vv+\cB) = \Z^n.\]
If $\cB$ both packs and covers $\Z^n$ by $T$, then we say $\cB$
\emph{tiles} $\Z^n$ by $T$.  The \emph{packing density} (or
\emph{covering density}, respectively) of $\cB$ by $T$ is defined as
\[ \delta\eqdef \lim_{\ell\to\infty} \frac{\abs*{[-\ell,\ell]^n\cap T}\cdot\abs{\cB}}{\abs*{[-\ell,\ell]^n}}.\]
When $T=\Lambda$ is some lattice, we call these lattice packings and
lattice coverings, respectively. The density then takes on a simpler form
\[\delta = \frac{\abs{\cB} }{\vol(\Lambda)}.\]

Throughout the paper, the object we pack and cover $\Z^n$ with, is the
error ball, $\BALL$, defined in \eqref{eq:ball}. We conveniently
observe that for all integers $n\geq 1$, $0\leq t\leq n$, $0\leq
\km\leq \kp$, we have
\[ \abs*{\BALL}=\sum_{i=0}^t \binom{n}{i}(\kp+\km)^i.\]

\subsection{Lattice Packing/Covering/Tiling and Group Splitting}

Lattice packing, covering, and tiling of $\Z^n$ with
$\cB(n,t,\kp,\km)$, in connection with group splitting, has a long
history when $t=1$ (e.g., see \cite{Ste67}), called lattice tiling by
crosses if $\kp=\km$ (e.g., \cite{Ste84}), semi-crosses when $\km=0$
(e.g., \cite{Ste84,HamSte84,HicSte86}), and quasi-crosses when
$\kp\geq\km\geq 0$ (e.g., \cite{Sch12,Sch14}). For an excellent
treatment and history, the reader is referred to~\cite{SteSza94} and
the many references therein. Other variations, keeping $t=1$
include~\cite{Tam98,Tam05}. More recent results may be found
in~\cite{YeZhaZhaGe20} and the references therein.

For $t\geq 2$, an extended definition of group splitting in connection
with lattice tiling is provided in \cite{WeiSchwartz}. In the
following, we modify this definition to distinguish between lattice
packings, coverings, and tilings.

\begin{definition}
  \label{def:split}
  Let $G$ be a finite Abelian group, where $+$ denotes the group
  operation. For $m\in\Z$ and $g\in G$, let $mg$ denote $g+g+\dots+g$
  (with $m$ copies of $g$) when $m>0$, which is extended in the
  natural way to $m\leq 0$. Let $M\subseteq\Z\setminus\set*{0}$ be a
  finite set, and $S=\set*{s_1,s_2,\dots,s_n}\subseteq G$. 
  \begin{enumerate}
  \item
    If the elements $\ve\cdot (s_1,\dots,s_n)$, where $\ve\in
    (M\cup\set{0})^n$ and $1\leq \wt(\ve)\leq t$, are all distinct and
    non-zero in $G$, we say the
    set $M$ partially $t$-splits $G$ with splitter set $S$, denoted
    \[G \geq M\splt_t S.\]
  \item
    If for every $g\in G$ there exists a vector $\ve\in
    (M\cup\set{0})^n$, $\wt(\ve)\leq t$, such that $g=\ve\cdot
    (s_1,\dots,s_n)$, we say the set $M$ completely $t$-splits $G$
    with splitter set $S$, denoted
    \[G \leq M\splt_t S.\]
  \item
    If $G \geq M\splt_t S$ and $G \leq M\splt_t S$ we say $M$
    $t$-splits $G$ with splitter set $S$, and write
    \[ G = M\splt_t S.\]
  \end{enumerate}
\end{definition}

In our context, since we are interested in packing and covering with
$\BALL$, then in the previous definition, we need to take $M\eqdef
[-\km,\kp]^*$. Thus, the following two theorems show the equivalence
of partial $t$-splittings with $M$ and lattice packings of $\BALL$,
summarizing Lemma 3 and Lemma 4 in\cite{BuzEtz12}, and similarly for
complete $t$-splittings and lattice coverings.

\begin{theorem}
  \label{th:lattotile}
  Let $G$ be a finite Abelian group, $M\eqdef [-\km,\kp]^*$, and
  $S=\set*{s_1,\dots,s_n}\subseteq G$. Define $\phi:\Z^n\to G$ as
  $\phi(\vx)\eqdef\vx\cdot(s_1,\dots,s_n)$ and let
  $\Lambda\eqdef\ker\phi$ be a lattice. 
  \begin{enumerate}
  \item
    If $G \geq M\splt_t S$, then $\cB(n,t,\kp,\km)$ packs $\Z^n$ by
    $\Lambda$.
  \item
    If $G \leq M\splt_t S$, then $\cB(n,t,\kp,\km)$ covers $\Z^n$ by
    $\Lambda$.
  \end{enumerate}
\end{theorem}

\begin{IEEEproof}
  For packing, see Lemma 4 in \cite{BuzEtz12}. For covering, denote
  $\cB\eqdef\cB(n,t,\kp,\km)$. Assume $\vx\in\Z^n$. Since $G \leq
  M\splt_t S$, there exists a vector $\ve\in\cB$ such that
  $\phi(\vx)=\phi(\ve)$. Then $\vv\eqdef \vx-\ve\in\Lambda$, and
  $\vx\in\vv+\cB$.
\end{IEEEproof}

In the theorem above, for $G \geq M\splt_t S$, since the quotient
group $\Z^n / \Lambda$ is isomorphic to the image of $\phi$, which is
a subgroup of $G$, we have $\vol(\Lambda)\leq \abs{G}$. Then the
packing density of $\Lambda$ is
\[\delta = \frac{\abs{\BALL} }{\vol(\Lambda)}\geq\frac{\abs{\BALL} }{\abs{G}} .\]
For $G \leq M\splt_t S$, $\vol(\Lambda)= \abs{G}$, and the
covering density of $\Lambda$ is
\[\delta = \frac{\abs{\BALL} }{\vol(\Lambda)}=\frac{\abs{\BALL} }{\abs{G}} .\]

It is known that a lattice packing implies a partial splitting. While
not of immediate use to us in this paper, we do mention that an
analogous claim is also true for lattice coverings, as the following theorem
shows.

\begin{theorem}
  Let $\Lambda\subseteq\Z^n$ be a lattice. Define $G\eqdef \Z^n /
  \Lambda$. Let $\phi:\Z^n\to G$ be the natural homomorphism, namely
  the one that maps any $\vx\in\Z^n$ to the coset of $\Lambda$ in
  which it resides, and then $\Lambda=\ker\phi$. Finally, let $\ve_i$
  be the $i$-th unit vector in $\Z^n$ and set $s_i\eqdef\phi(\ve_i)$ for
  all $1\leq i\leq n$ and $S\eqdef\set*{s_1,s_2,\dots,s_n}$.
  \begin{enumerate}
  \item If $\cB(n,t,\kp,\km)$ packs $\Z^n$ by $\Lambda$, then  $G\geq M\splt_t S$;
  \item  if $\cB(n,t,\kp,\km)$ covers $\Z^n$ by $\Lambda$, then  $G\leq M\splt_t S$,
  \end{enumerate}
  where $M\eqdef [-\km,\kp]^*$.
\end{theorem}
\begin{IEEEproof}
  The packing see Lemma 3 in \cite{BuzEtz12}. Now we prove the claim
  for covering. Let $\Lambda+\vx \in G$ be any element of $G$. Since
  $\BALL$ covers $\Z^n$ by $\Lambda$, there exist $\vv\in\Lambda$ and
  $\ve\in\BALL$ such that $\vx=\vv+\ve$. This means
  \[\Lambda+\vx=\phi(\vx)=\phi(\vv)+\phi(\ve)=\phi(\ve)=\ve\cdot(s_1,\dots,s_n),\]
  which completes the proof.
\end{IEEEproof}

Finally, a connection between perfect codes in the Hamming metric, and
lattice tilings with $\BALL$ was observed in~\cite{WeiSchwartz}. We repeat
a theorem we shall later generalize.

\begin{theorem}[Theorem~3 in \cite{WeiSchwartz}]
  \label{th:perfectcode}
  In the Hamming metric space, let $C$ be a perfect linear
  $[n,k,2t+1]$ code over $\F_p$, with $p$ a prime. If $\kp+\km+1=p$,
  then
  \[ \Lambda\eqdef \set*{ \vx\in\Z^n ; (\vx\bmod p)\in C}\]
  is a lattice, and $\cB(n,t,\kp,\km)$ tiles $\Z^n$ by
  $\Lambda$.
\end{theorem}


\section{Constructions of Lattice Packings}
\label{sec:packing}

In this section we describe several constructions for packings of
$\BALL$. We begin by showing how to translate codes in the Hamming
metric into lattices that pack $\BALL$. Apart from a single case,
these have vanishing density. The motivation for showing these
``off-the-shelf'' constructions is to create a baseline against which
we measure our tailor-made constructions that appear later. These use
$B_t[N;1]$ sequences, or take inspiration from constructions of sets
with no arithmetic progression, to construct codes that improve upon
the baseline.

\subsection{Constructions Based on Error-Correcting Codes}

Theorem~\ref{th:perfectcode} can be easily modified to yield the
following construction, the proof of which is the same as that of
\cite[Theorem 3]{WeiSchwartz} and we omit here to avoid unnecessary
repetition.

\begin{theorem}
  \label{th:linearcode}
  In the Hamming metric space, let $C$ be a linear $[n,k,2t+1]$ code
  over $\F_p$, with $p$ a prime. If $0\leq \kp+\km < p$ are integers,
  then
  \[ \Lambda\eqdef \set*{ \vx\in\Z^n ; (\vx\bmod p)\in C}\]
  is a lattice, and $\cB(n,t,\kp,\km)$ packs $\Z^n$ by
  $\Lambda$.
\end{theorem}

Since $\cB(n,t,\kp,\km)$ packs $\Z^n$ by $\Lambda$, the lattice
$\Lambda$ is an error-correcting code over $\Z$ for asymmetric
limited-magnitude errors.  We note that a similar construction of
error-correcting codes over a finite alphabet for asymmetric
limited-magnitude errors was presented in \cite{CasSchBohBru10} and
the decoding scheme therein can be adapted here as follows. Let
$\vx\in\Lambda$ be a codeword, and $\vy \in\vx+\BALL$ be the channel
output.  Denote $\boldsymbol{\psi}= \vy \pmod{p}$. Run the decoding
algorithm of the linear $[n,k,2t+1]$ code on $\boldsymbol{\psi}$ and
denote the output as $\boldsymbol{\phi}$. Then $\boldsymbol{\phi}$ is
a codeword of the linear code over $\F_p$ and it is easy to see that
$\boldsymbol{\phi}=\vx \pmod{p}$. Thus $\vy-\vx\equiv\boldsymbol{\psi}
- \boldsymbol{\phi} \pmod{p}$. Denote
$\boldsymbol{\epsilon}=\boldsymbol{\psi} - \boldsymbol{\phi} \pmod{p}$
and let $\ve=(e_1,e_2,\ldots,e_n)$ where
\[
e_i \eqdef \begin{cases} \epsilon_i, \textup{\ \ if $0\leq \epsilon_i \leq \kp$;} \\
\epsilon_i-p, \textup{\ \ otherwise.}\\
\end{cases}
\]
Then $\vx$ can be decoded as $\vx=\vy-\ve$.

Now let us look at the packing density.
  
\begin{corollary}
  \label{cor:linearcode}
Let $\Lambda$ be the lattice constructed in
Theorem~\ref{th:linearcode}. Then $\vol(\Lambda)=p^{n-k}$
and the packing density is
\[\delta=\frac{\sum_{i=0}^t \binom{n}{i}(\kp+\km)^i}{ p^{n-k} }.\]
\end{corollary}
\begin{IEEEproof}
  In Theorem~\ref{th:linearcode}, the quotient group $\Z^n/ \Lambda$
  is isomorphic to the group $\Z_p^{n-k}$ (see Example 3 and Example 4
  in \cite{WeiSchwartz}). The claim is then immediate.
\end{IEEEproof}

When $t$ is small, we may use BCH codes as the input to construct the
lattice packing.

\begin{theorem}[Primitive narrow-sense BCH codes {\cite[Theorem 10]{aly2007quantum}}]\label{primitiveBCH} Let $p$ be a prime. 
Fix $m\ge 1$ and $2\le d\le p^{\ceil{m/2}}-1$. Set $n=p^m-1$.  Then
there exists an $[n,k,d]$-code over $\F_p$ with
\[k=n-\ceil{(d-1)(1-1/p)}m.\]
\end{theorem}

\begin{corollary}\label{cor:lc}
Let $\psi(x)$ be the smallest prime not smaller than $x$ \footnote{It
  is known that $\psi(x)\leq x+x^{21/40}$ \cite{BakHarPin01}, and
  conjectured that $\psi(x)=x+O(\log x)$.} and denote
$p\eqdef\psi(\kp+\km+1)$. Let $m,t$ be positive integers such that
$2t\le p^{\ceil{m/2}}-2$, and set $n=p^m-1$. Then $\Z^n$ can be
lattice packed by $\BALL$ with density
\[\delta = \frac{\sum_{i=0}^t \binom{n}{i}(\kp+\km)^i}{(n+1)^{\ceil{2t(1-{1}/p)}} }.\]
\end{corollary}
\begin{IEEEproof}
  Simply combine Theorem~\ref{primitiveBCH} with
  Corollary~\ref{cor:linearcode}.
\end{IEEEproof}

Note that if $\kp=1$ and $\km=0$, then the packing density in
Corollary~\ref{cor:lc} is $\delta=\frac{\sum_{i=0}^t
  \binom{n}{i}}{(n+1)^{t}}=\frac{1}{t!}+o(1)$ (when $t$ is fixed and
$n$ tends to infinity). However, for all the other values of $\kp$ and
$\km$, namely $p\geq 3$, the density always vanishes when $n$ tends to
infinity, i.e., $\delta=\Theta(n^{t-\ceil{2t(1-{1}/p)}})$. In the remainder
of this section, we will present some constructions to provide lattice
packings of higher density.

Perfect codes were used in \cite{WeiSchwartz} obtain lattice tilings,
i.e., lattice packings with density $1$. Similarly, it is possible to
use quasi-perfect linear codes to obtain lattice packings with high
densities.

\begin{corollary}\label{cor:quasiperfectlc} Assume that $1\leq \kp+\km\leq 2$
  are non-negative integers.
  \begin{enumerate}
  \item
    Let $m$ be a positive integer and $n=(3^m+1)/2$. Then $\Z^n$ can
    be lattice-packed by $\cB(n,2,\kp,\km)$ with density
    \[\delta = \frac{  {n\choose 2}(\kp+\km)^2 +n(\kp+\km) +1 }{(2n-1)^2}.\]
  \item
    Let $m\geq 3$ be an odd integer and $n=(3^m-1)/2$. Then $\Z^n$ can
    be lattice packed by $\cB(n,2,\kp,\km)$ with density
    \[\delta = \frac{  {n\choose 2}(\kp+\km)^2 +n(\kp+\km) +1 }{(2n+1)^2}.\]
  \end{enumerate}
\end{corollary}
\begin{IEEEproof}
For the first case, we take a $[(3^m+1)/2, (3^m+1)/2-2m,5]_3$ code
from \cite{GashkovSidelnikov1986} as the input of
Theorem~\ref{th:linearcode} to obtain the lattice packing, while for
the second, we take a $[(3^m-1)/2, (3^m-1)/2-2m,5]_3$ code from
\cite{DanevDodunekov2008} as the input.
\end{IEEEproof}

We note that \cite{LiHelleseth2016} presented some binary
quasi-perfect linear codes with minimum distance $5$, which can give
rise to packings of $\cB(n,2,1,0)$. It has been checked out that the
corresponding densities are asymptotically the same as that in
Corollary~\ref{cor:lc}, i.e.,
$\frac{1}{2}+o(1)$. \cite{LiHelleseth2016} also studied $p$-ary
quasi-perfect linear codes with $p\geq 3$. However, the minimum
distances of those codes are no more that $4$. So they cannot be used
to obtain packings of $\BALL$ with $t \geq 2$.

The following theorem uses non-linear codes to construct non-lattice
packing, the proof of which is the same as that of \cite[Theorem
  5]{CasSchBohBru10} and we omit here to avoid unnecessary repetition.

\begin{theorem}
  \label{th:nonlinearcode}
  In the Hamming metric space, let $C$ be a $q$-ary $(n,M,2t+1)$
  code. Denote
  \[ V\eqdef \set*{ \vv\in\Z^n ; (\vv\bmod q)\in C}.\]
  If $\kp+\km < q$, then for any distinct $\vv,\vv'\in V$, we have
  $(\vv+\BALL) \cap (\vv'+\BALL)=\varnothing$, namely, $\BALL$ can
  pack $\Z^n$ by $V$.
\end{theorem}

\begin{corollary}
  \label{cor:nonlinearcode}
  The density of the packing of $\Z^n$ constructed in
  Theorem~\ref{th:nonlinearcode} is
  \[\delta=\frac{M\cdot\sum_{i=0}^t \binom{n}{i}(\kp+\km)^i}{ q^{n} }.\]
\end{corollary}
\begin{IEEEproof}
  Note that the set $V$ constructed in Theorem~\ref{th:nonlinearcode}
  has period $q$ in each coordinate. Thus, the packing density of
  $\Z^n$ is equal to the packing density of $\Z_q^n$, which is
  $\frac{M}{q^n}\sum_{i=0}^t \binom{n}{i}(\kp+\km)^i.$
\end{IEEEproof}

\begin{corollary}\label{cor:perparata}
  Let $m \geq 4$ be even integer and let $n=2^m-1$. Then $\Z^n$ can be
  packed by $\cB(n,2,1,0)$ with density
  \[\delta = \frac{ {n\choose 2} +n +1 }{(n+1)^2/2}=1-\frac{n-1}{(n+1)^2}.\]
\end{corollary} 
\begin{proof}
  We take a binary $(2^m-1, 2^{2^m-2m},5)$ Preparata code
  \cite{Preparata1968} as the input of Theorem~\ref{th:nonlinearcode}
  to obtain the packing.
\end{proof}

\subsection{Construction Based on $B_t[N;1]$ Sets for $(\kp,\km)=(1,0)$ or $(1,1)$}

A subset $A\subseteq\Z$ is called a \emph{$B_t[g]$ set} if every
integer can be written in at most $g$ different ways as a sum of $t$
(not necessary distinct) elements of $A$ (e.g., see
\cite[Section~4.5]{TaoVu06} and the many references therein). In this
section, however, we require $B_t[1]$ sets with a somewhat stronger
property. Specifically, a subset $A$ of $\Z_N$ is called a
\emph{$B_t[N;1]$ set} if the sums of any $t$ (not necessary distinct)
elements of $A$ are all different modulo $N$. Bose and Chowla
\cite{BoseChowla} presented two classes of $B_t[N;1]$ sets.

\begin{theorem}\cite[Theorem 1 and Theorem 2]{BoseChowla}\label{thm:Btset}
Let $q$ be a prime power and $t$ be a positive integer.  Let
$\alpha_0=0,\alpha_1,\alpha_2,\ldots,\alpha_{q-1}$ be all the
different elements of $\F_q$.
\begin{enumerate} 
\item
  Let $\xi$ be a primitive element of the extended field
  $\F_{q^t}$. For each $0\leq i \leq q-1$, let $d_i \in \Z_{q^{t}-1}$
  be such that
  \[\xi^{d_i}=\xi+\alpha_i.\]
  Then the set $S_1\eqdef \set*{d_i;0\leq i \leq q-1}$ is a
  $B_t[q^t-1;1]$ set of size $q$.
\item
  Let $\eta$ be a primitive element of the extended field
  $\F_{q^{t+1}}$. For each $0\leq i \leq q-1$, let $\beta_i \in \F_q$
  and $s_i \in \Z_{(q^{t+1}-1)/(q-1)}$ such that
  \[ \beta_i\eta^{s_i}=\eta+\alpha_i.\]
  Then the set $S_2\eqdef \set*{s_i;0\leq i \leq q-1} \cup \set{0}$ is a
  $B_t[(q^{t+1}-1)/(q-1);1]$ set of size $q+1$.
\end{enumerate}
\end{theorem}

\begin{theorem}\label{thm:con-Btset-1}
Let $A$ be a $B_t[N;1]$ set which contains $0$. Denote  $S\eqdef A\setminus\set{0}$. Then   $\Z_N \ge \set{1} \splt_t S$. 
\end{theorem}

\begin{IEEEproof}
Suppose to the contrary that
$\set*{s_{i_1},s_{i_2},\ldots,s_{i_{\ell}}}$ and
$\set*{s_{j_1},s_{j_2},\ldots,s_{j_{r}}}$ are two distinct
subsets of $S$ such that
\[s_{i_1}+s_{i_2}+\ldots+s_{i_{\ell}}\equiv s_{j_1}+s_{j_2}+\ldots+s_{j_{r}} \pmod{N},\]
where $\ell,r\leq t$.
Then we have 
\[ \underbrace{0+0+\cdots +0}_{t-\ell}+ s_{i_1}+s_{i_2}+\ldots+s_{i_{\ell}}\equiv \underbrace{0+0+\cdots +0}_{t-r}+s_{j_1}+s_{j_2}+\ldots+s_{j_{r}} \pmod{N},\]
which contradicts that   $S\cup \set{0}$ is a $B_t[N;1]$ set.
\end{IEEEproof}

The following result slightly improves upon the density obtained in
Corollary~\ref{cor:lc} for lattice packings of $\cB(n,t,1,0)$.

\begin{corollary}
  \label{cor:btsets}
  \begin{enumerate}
  \item
    If $n+1$ is a prime power, then for any $2\leq t \leq n$, there is
    a lattice packing of $\Z^n$ by $\cB(n,t,1,0)$ with density
    \[\delta=\frac{\sum_{i=0}^t\binom{n}{i} }{(n+1)^t-1}=\frac{1}{t!}+o(1).\]
  \item
    If $n$ is a prime power, then for any $2\leq t \leq n$, there is a
    lattice packing of $\Z^n$ by $\cB(n,t,1,0)$ with
    density
    \[\delta=\frac{\sum_{i=0}^t\binom{n}{i}}{(n^{t+1}-1)/(n-1)}=\frac{1}{t!}+o(1).\]
  \end{enumerate}
\end{corollary}
\begin{IEEEproof}
  Note that if $\set*{s_1,s_2,\ldots,s_n}$ is a $B_t[N;1]$ set, then
  $\set*{0,s_2-s_1,s_3-s_1,\ldots,s_n-s_1}$ is also a $B_t[N;1]$ set,
  which contains $0$. Hence, combining Theorem~\ref{thm:Btset} and
  Theorem~\ref{thm:con-Btset-1}, together with
  Theorem~\ref{th:lattotile}, we prove the claim.
\end{IEEEproof}

In general, we do not have an efficient decoding scheme for the
lattice code obtained from Theorem~\ref{thm:con-Btset-1}. However, for
the lattice code $\Lambda_{S_2\setminus\set{0}}$ obtained from the
$B_t[(q^{t+1}-1)/(q-1);1]$ set $S_2$ in Theorem~\ref{thm:Btset}, we
have the following decoding algorithm (summarized in
Algorithm~\ref{alg:decoding}). Let $n=q$ and let
$S_2\setminus\set{0}=\set{s_0,\dots,s_{n-1}}$ be defined as in
Theorem~\ref{thm:Btset}. Let $\vx\in\Lambda_{S_2\setminus\set{0}}$ be
a codeword and $\vy \in \vx+\cB(n,t,1,0)$ be the channel output.  Then
$\vx\cdot (s_0,s_1,\ldots,s_{n-1})=0$, and $\vy-\vx$ is a binary
vector over $\set{0,1}$ of weight at most $t$. Let
$i_1,i_2,\ldots,i_r$ be the indices of the nonzero bits of $\vy-\vx$,
and denote $s=\vy\cdot (s_0,s_1,\ldots,s_{n-1})$. We aim to recover
$i_1,i_2,\ldots,i_r$ from $s$.  Since
\begin{align*}
s&=\vy\cdot (s_0,s_1,\ldots,s_{n-1}) = \vy\cdot (s_0,s_1,\ldots,s_{n-1}) -0 =\vy\cdot (s_0,s_1,\ldots,s_{n-1})-\vx\cdot (s_0,s_1,\ldots,s_{n-1}) \\
&= (\vy-\vx)\cdot (s_0,s_1,\ldots,s_{n-1}) =\sum_{\ell=1}^r s_{i_\ell},
\end{align*}
we have that 
\begin{equation}\label{eq:decode}
\parenv*{\prod_{\ell=1}^r \beta_{i_\ell}}  \eta^s=  \prod_{\ell=1}^r\parenv*{\beta_{i_\ell} \eta^{s_{i_\ell}}} =\prod_{\ell=1}^r(\eta+\alpha_{i_\ell}).
\end{equation}
Let $p(x)$ be the primitive polynomial of $\eta$ and $r(x) = x^s \bmod p(x)$. Then $r(\eta)=\eta^s$, and we substitute this in \eqref{eq:decode} to obtain
\[\parenv*{\prod_{\ell=1}^r \beta_\ell}  r(\eta)=\prod_{\ell=1}^r(\eta+\alpha_{i_\ell}).\]
Since both the polynomials $(\prod_{\ell=1}^r \beta_{i_\ell}) r(x)$
and $\prod_{\ell=1}^r(x+\alpha_{i_\ell})$ are over $\F_q$ and have
degrees at most $t$, they should be the same; otherwise, $\eta$ is a
root of a nonzero polynomial of degree at most $t$, which contradicts
the fact that $\eta$ is a primitive element of $\F_{q^{t+1}}$. Thus,
we may solve $(\prod_{\ell=1}^r \beta_\ell) r(x)$ to find out
$\alpha_{i_1}, \alpha_{i_1},\ldots, \alpha_{i_r}$. Finally, we can subtract
$\sum_{\ell=1}^r \ve_{i_\ell}$ from $\vy$ to obtain $\vx$.

\begin{algorithm}[ht]
  \caption{Decoding algorithm for $\Lambda_{S_2\setminus\set{0}}$ from Theorem~\ref{thm:Btset}}
  \label{alg:decoding}
  \begin{algorithmic}[1]
    \Statex \textbf{Input:} received vector $\vy\in\Z^n$ suffering at most $t$ errors
    \Statex \quad\quad\quad $S_2\setminus\set{0}=\set{s_0,\dots,s_{n-1}}$ from Theorem~\ref{thm:Btset}
    \Statex \quad\quad\quad where $\eta$ is a root of a primitive polynomial $p(x)$ of degree $t+1$
    \Statex \quad\quad\quad and where $\F_q\setminus\set{0}=\set{\alpha_1,\dots,\alpha_{q-1}}$.
    \Statex \textbf{Output:} codeword $\vx\in\Lambda_{S_2\setminus\set{0}}$ such that $\vy\in\vx+\cB(n,t,1,0)$
    \State $s \gets \vy \cdot (s_0,s_1,\ldots,s_{n-1})$ \label{step:s}
    \State $r(x)\gets x^s \bmod p(x)$ \label{step:r}
    \For{$1\leq i\leq q-1$} \label{step:loop}
    \If{$r(\alpha_i)=0$}
    \State $\vy \gets \vy - \ve_i$
    \EndIf
    \EndFor
    \State \Return{$\vy$}
  \end{algorithmic}
\end{algorithm}

Let us analyze the time complexity of Algorithm~\ref{alg:decoding},
where we count the number of field operations in $\F_q$. The inner
product in Step~\ref{step:s} takes $O(n)$
operations. Step~\ref{step:r} is possible to compute in $O(t^2 \log
s)$ field operations (by using successive squaring and multiplication
by $x$ as necessary, taking a modulo $p(x)$ after each
iteration). Since $s \in \Z_{(q^{t+1}-1)/(q-1)}$ and $n=q$, it is
$O(t^3 \log n)$. Finally, the root search loop starting in
Step~\ref{step:loop} takes $O(tn)$ operations. Thus, in total, the
time complexity $O(tn+t^3\log n)$ field operations. If $t$ is
constant, then this is linear in the code length. As a final comment,
we point out that $q=O(n)$, and thus the basic field operations of
addition and multiplication may be realized in $O(\mathrm{polylog}(n))$ time.

We now move from packing $\cB(n,t,1,0)$ to packing $\cB(n,t,1,1)$. In
general, we note that one can use a $B_{h}[N;1]$ set with
$h=t(\kp+\km)$ to obtain a lattice packing of $\BALL$ in $\Z_n$ for
$\kp+\km \geq 2$. However, in this case, the density is
$O(n^{t(1-{\kp+\km})})$, which vanishes when $n$ tends to
infinity. Similarly, the lattice packing from Corollary~\ref{cor:lc}
also has vanishing density $O(n^{t-\ceil{2t(1-{1}/p)}})$.  In the
following, we give a modified construction which uses a $B_{t}[N;1]$
set to obtain a lattice packing of $\cB(n,t,1,1)$ with density
$\Omega(1)$.

\begin{theorem}\label{thm:con-Btset-pm1}
Let $A=\set{a_1,a_2,\dots,a_n}$ be a $B_t[N;1]$ set. In the group
$\Z_N \times \Z_{2t+1}$, construct a set
\[S\eqdef \set*{ (a_i,1) ; a_i \in A}.\]
Then $\Z_N \times \Z_{2t+1} \ge \set{-1,1} \splt_t S$. 
\end{theorem}

\begin{IEEEproof}
Suppose to the contrary that there are
$(a_{i_1},1),(a_{i_1},1),\dots,(a_{i_{\ell}},1)$ and
$(a_{j_1},1),(a_{j_2},1),\dots,(a_{j_{r}},1)$ in $S$ such that
\begin{equation}\label{eq:con-Bset}
\sum_{m=1}^{\ell'} (a_{i_m},1)-\sum_{m=\ell'+1}^{\ell}(a_{i_m},1) =
\sum_{m=1}^{r'} (a_{j_m},1)-\sum_{m=r'+1}^{r}(a_{j_m},1),
\end{equation}
where $0\leq \ell'\leq \ell \leq t$ and $0\leq r'\leq r \leq t$, and
the addition is over the group $\Z_N \times \Z_{2t+1}$.

The second coordinate of the equation above implies that
\[\ell-2\ell'\equiv r-2r' \pmod{2t+1}.\]
Since $0\leq \ell'\leq \ell \leq t$ and $0\leq r'\leq r \leq t$, we
have $\ell-2\ell', r-2r'\in [-t,t]$. It follows that $\ell-2\ell'=
r-2r',$ and so, $\ell'+r-r'=r'+\ell-\ell'$. Let
$\tau\eqdef\ell'+r-r'$. Then
\[\tau = \frac{\ell'+r-r'+r'+\ell-\ell'}{2} =\frac{\ell+r}{2} \leq t.\]
Rearranging the terms in the first coordinate of the equation
\eqref{eq:con-Bset}, we have
\begin{equation}\label{eq:con-Bset-2}
a_{i_1}+ a_{i_2}+\cdots+ a_{i_{\ell'}} + a_{j_{r'+1}} + \cdots + a_{j_{r}} \equiv  a_{j_1}+ a_{j_2}+\cdots+ a_{j_{r'}} + a_{i_{\ell'+1}} + \cdots + a_{i_{\ell}} \pmod{N}.
\end{equation}
On both side of \eqref{eq:con-Bset-2}, there are $\tau$ terms. This contradicts the fact that $A$ is a $B_{t}[N;1]$ (and hence  a $B_{\tau}[N;1]$ set for any $\tau \leq t$).
\end{IEEEproof}

Combining Theorem~\ref{thm:Btset} and Theorem~\ref{thm:con-Btset-pm1}, together with Theorem~\ref{th:lattotile}, we have the following result.

\begin{corollary}
  \label{cor:kpm1}
  If $n$ is a prime power, then for any $2\leq t \leq n$, there is a
  lattice packing of $\Z^n$ by $\cB(n,t,1,1)$ with density
  \[\delta=\frac{\sum_{i=0}^t\binom{n}{i}2^i }{(2t+1)(n^t-1)}=\frac{2^t}{t!(2t+1)}+o(1).\] 
  If $n-1$ is a prime power, then for any $2\leq t \leq n$, there is a
  lattice packing of $\Z^n$ by $\cB(n,t,1,1)$ with density
  \[\delta=\frac{\sum_{i=0}^t\binom{n}{i} 2^i}{(2t+1)((n-1)^{t+1}-1)/(n-2)}=\frac{2^t}{t!(2t+1)}+o(1).\]
\end{corollary}

\subsection{Constructions for $t=2$}

Whereas in the previous section we considered unconstrained $t$ but
only small values of $\kp,\km$, in this section we focus on the case
of $t=2$ but unconstrained $\kp,\km$.

We first present a construction based on $k$-fold Sidon sets. Such
sets were first defined in~\cite{Lazebnik2003} as a generalization of
Sidon sets.  We repeat the definition here. Let $k$ be a positive
integer and let $N$ be relatively prime to all elements of $[1,k]$,
i.e., $\gcd(N,k!)=1$. Fix integers $c_1,c_2,c_3,c_4\in [-k,k]$ such
that $c_1+c_2+c_3+c_4=0$, and let $\cS$ be the collection of sets $S
\subseteq \set{1,2,3,4}$ such that $\sum_{i \in S} c_i=0$ and $c_i\neq
0$ for $i \in S$. We note that $\cS$ always contains the empty
set. Consider the following equation over $x_1,x_2,x_3,x_4\in \Z_N$:
\begin{equation}\label{eq:kfoldSidonset}
c_1x_1+c_2x_2+c_3x_3+c_4x_4\equiv 0 \pmod{N}.
 \end{equation}
A solution of \eqref{eq:kfoldSidonset} is \emph{trivial} if there
exists a partition of the set $\set{i ; c_i\neq 0}$ into sets $S,T\in
\cS$ such that $x_i=x_j$ for all $i,j\in S$ and all $i,j\in T$. We now
define a \emph{$k$-fold Sidon set} to be a set $A\subseteq\Z_N$ such
that for any $c_1,c_2,c_3,c_4\in [-k,k]$ with $c_1+c_2+c_3+c_4=0$,
equation \eqref{eq:kfoldSidonset} has only trivial solutions in
$A$. In the special case of $k=1$, a $1$-fold Sidon set coincides with
the usual definition of a \emph{Sidon set}, which is also a $B_2[N;1]$
set.

\begin{theorem}\label{thm:con-kfoldSidonset}
Let $A=\set{a_1,a_2,\dots,a_n}\subseteq\Z_n$ be a $k$-fold Sidon
set. Assume that $0\leq \km\leq \kp\leq k$ and $\kp+\km\geq 1$. In the
group $G \triangleq \Z_{2(\kp+\km)+1} \times \Z_N$, construct a set
\[S\eqdef \set*{ (1,x) ; x \in A}.\]
Then $G \ge [-\km,\kp]^* \splt_2 S$. 
\end{theorem}

\begin{IEEEproof}
Suppose to the contrary that $G$ is not partially $2$-split by
$S$. Then there are $x_1,x_2,x_3,x_4 \in A$ and $c_1,c_2,c_3,c_4 \in
[-\km,\kp]$ such that
\[ c_1+c_2\equiv c_3+c_4 \pmod{2(\kp+\km)+1},\]
and 
\begin{equation}\label{eq:con-kfoldSidonset}
c_1 x_1+c_2 x_2\equiv c_3 x_3 + c_4 x_4 \pmod{N},
\end{equation}
where all the following hold:
\begin{enumerate}
\item
  $x_1\neq x_2$
\item
  $x_3\neq x_4$
\item
  $x_1\neq x_3$ if $c_1=c_3$ and $c_2=c_4=0$
\item
  $x_2\neq x_4$ if $c_2=c_4$ and $c_1=c_3=0$
\item
  $(x_1,x_2)\neq (x_3,x_4)$ if $(c_1,c_2)=(c_3,c_4)$.
\item
  $(x_1,x_2)\neq (x_4,x_3)$ if $(c_1,c_2)=(c_4,c_3)$.
\end{enumerate}
Since $-(\kp+\km) \leq a+b,c+d \leq \kp+\km$, it follows that
$c_1+c_2=c_3+c_4$, or equivalently, $c_1+c_2-c_3-c_4=0$.  To avoid
contradicting the assumption that $A$ is a $k$-fold Sidon set,
$(x_1,x_2,x_3,x_4)$ should be a trivial solution of
\eqref{eq:con-kfoldSidonset}. We consider the following cases:

\textbf{Case 1.}  If none of $c_1,c_2,c_3,c_4$ are $0$, we consider
the possible partitions of $\set{1,2,3,4}$. Since $x_1\neq x_2$ and
$x_3\neq x_4$, $1$ and $2$, respectively $3$ and $4$, cannot be placed
in the same set in the partition. Then the possible partitions are
$\set{\set{1,3}, \set{2,4}}$, and $\set{\set{1,4}, \set{2,3}}$. If the
partition is $\set{\set{1,3}, \set{2,4}}$, then $c_1-c_3=0$ and
$c_2-c_4=0$. It follows that $x_1=x_3$ and $x_2=x_4$, which
contradicts that $(x_1,x_2)\neq (x_3,x_4)$ when
$(c_1,c_2)=(c_3,c_4)$. The case of $\set{\set{1,4},\set{2,3}}$ is
proved symmetrically.

\textbf{Case 2.}  If there is exactly one element of $c_1,c_2,c_3,c_4$
that is equal to $0$, say w.l.o.g., $c_1=0$, then the only possible
partition of $\set{2,3,4}$ is $\set{ \emptyset, \set{2,3,4}}$, which
contradicts $x_3\neq x_4$.

\textbf{Case 3.}  If there are exactly two elements of
$c_1,c_2,c_3,c_4$ that are equal to $0$, w.l.o.g., we may consider the
two cases where $c_1=c_2=0$, and $c_1=c_3=0$. If $c_1=c_2=0$, the only
possible partition of $\set{3,4}$ is $\set{ \emptyset, \set{3,4} }$,
which contradicts $x_3\neq x_4$. If $c_1=c_3=0$, the only possible
partition of $\set{2,4}$ is $\set{ \emptyset, \set{2,4} }$. Then we
have $x_2=x_4$. Note that $c_1=c_3=0$ and $c_2=c_4$, and we get a
contradiction.

\textbf{Case 4.}  If there are exactly three elements of
$c_1,c_2,c_3,c_4$ that are equal to $0$, assume w.l.o.g., that
$c_2=c_3=c_4=0$ and $c_1\neq 0$. Then we need to partition
$\set{1}$. However, such a partition does not exist as $c_1\neq
0$. Thus, there is no solution to (\refeq{eq:con-kfoldSidonset}).
\end{IEEEproof}

When $k=2$, a family of $2$-fold Sidon sets is constructed in \cite{Lazebnik2003} by removing some elements from Singer difference sets with multiplier $2$. 

\begin{theorem}[Theorem~2.5 in \cite{Lazebnik2003}]\label{thm:2foldSidonset}
Let $m$ be a positive integer and $N=2^{2^{m+1}}+2^{2^m}+1$. Then
there exists a $2$-fold Sidon set $A\subseteq \Z_N$ such that
\[\abs{A} \geq \frac{1}{2} N^{1/2}-3. \] 
\end{theorem}

We immediately get the following corollary.

\begin{corollary}
  \label{cor:t2sidon}
  Let $0\leq \km\leq \kp\leq 2$ be integers with $\kp+\km\geq
  1$. There is an infinite family of integers $n$ such that $\Z^n$ can
  be lattice packed by $\cB(n,2,\kp,\km)$ with density
  \[\delta=\frac{\binom{n}{2}(\kp+\km)^2 + n(\kp+\km)+1 }{(2(\kp+\km)+1)(2n+6)^2}=\frac{1}{8(2(\kp+\km)+1)}+o(1).\] 
\end{corollary}
\begin{IEEEproof}
Simply combine Theorem~\ref{thm:con-kfoldSidonset} and
Theorem~\ref{thm:2foldSidonset}.
\end{IEEEproof}

Now, we present a construction for $t=2$ and $0\leq \km\leq \kp\leq
3$, which combines Behrend's method \cite{Behrend46} and Ruzsa's
method \cite{Ruzsa93} to forbid some specified linear equations.

\begin{theorem}\label{thm:con-t=2}
  Let $0\leq \km\leq \kp\leq 3$ be integers such that $\kp+\km\geq 1$.  Set
  $\alpha \eqdef \max\set{2\kp^2,3}$. Let $D \geq 2$ and $K\geq 1$ be integers,
  and $p\equiv\pm 5\pmod{12}$ be a prime such that $(\alpha K+1)^D
  \leq p$. For each $0\leq m < DK^2$, define
\[C_m\eqdef  \set*{ x= \sum_{i=0}^{D-1} x_i (\alpha K+1)^i ; 0\leq x_i\leq K, \sum_{i=0}^{D-1}x_i^2=m }. \]
Let $G \eqdef  \Z_{3\kp+2\km+1} \times \Z_p\times \Z_p$, and construct a subset 
\[ S_m\eqdef \set*{ s_x ; x \in C_m }, \textup{ where  }  s_x \eqdef (1,x,x^2) \in G.\]
If $\kp\leq 3$, then $G\geq M\splt_2 S_m$ for every $0\leq m < DK^2$,
and where $M\eqdef [-\km,\kp]^*$.
\end{theorem}

\begin{IEEEproof}
  Suppose to the contrary that
  $G$ is not partially $2$-split by $S_m$. We consider the following
  cases.

\textbf{Case 1.} $a s_x =\Zero$ for some $a \in M$ and $x\in C_m$. The
first coordinate of this equation is $a\equiv 0
\pmod{3\kp+2\km+1}$. Since $-\km\leq a \leq \kp$, necessarily $a=0$, a
contradiction.

\textbf{Case 2.} $a s_x =b s_y$ for some $a,b \in M$, $x,y \in C_m$
and $(a,x)\neq (b,y)$. Similarly to Case 1, the first coordinate implies
that $a=b$. From the second coordinate, we have $a x \equiv b y
\pmod{p}$, and so, $x\equiv y \pmod{p}$. It follows that $x =y$ as
$0\leq x,y <p$, which contradicts $(a,x)\neq (b,y)$.
 
\textbf{Case 3.} $a s_x + b s_y = \Zero$ for some $a,b \in M$, $x,y
\in C_m$ and $x\neq y$. The first coordinate implies $a+b=0$ as $-2\km
\leq a+b\leq 2\kp$. W.l.o.g., we assume $a>0$. Then $a s_x= (-b) s_y$,
where $0< a, -b \leq \km$, which was ruled out in Case 2.

\textbf{Case 4.}  $a s_x + b s_y = c s_u$ for some $a,b,c \in M$ and
$x,y,u\in C_m$ with $x\not =y$. From the first coordinate, we have
$a+b=c$ as $-2\km \leq a+b\leq 2\kp$ and $-\km\leq c \leq \kp$.  If
$x=u$ or $y=u$, then $b s_y=bs_u$ or $a s_x=as_u$, respectively, both
of which were ruled out in Case 2. Thus, in the following, we assume
$x,y,u$ are pairwise distinct. Furthermore, using the condition
$a+b=c$ and rearranging the terms, we may assume that $a,b,c>0$.

From the second coordinate, we have that $a x + by \equiv c u \pmod{p}$, or equivalently,
\[  \sum_{i=0}^{D-1} (ax_i+b y_i) (\alpha K+1)^i \equiv  \sum_{i=0}^{D-1} c u_i (\alpha K+1)^i   \pmod{p}.\]
Note that $0\leq ax_i+b y_i, cu_i\leq 2\kp K <\alpha K+1$ and $p \geq
(\alpha K+1)^D$.  It follows that $ax_i+b y_i=cu_i$ for all $0\leq
i\leq D-1$. Thus, the three distinct points $(x_0,x_1,\ldots,
x_{D-1})$, $(y_0,y_1,\ldots, y_{D-1})$, and $(u_0,u_1,\ldots,
u_{D-1})$, are collinear in $\Z^D$ where $D\geq 2$, which contradicts
the fact that they are on the same sphere, i.e., $\sum_i x_i^2=\sum_i
y_i^2=\sum_i u_i^2=m$.

\textbf{Case 5.}  $a s_x + b s_y = c s_u +d s_v$ for some $a,b,c,d \in
M$, $x,y,u,v\in C_m$, $x\neq y$ and $u\neq v$, where $abcd$ is
negative. By rearranging the terms, we may assume w.l.o.g.  that
\[a s_x + b s_y  + c s_z = d s_u\] for some $0<a,b,c,d \leq \kp$ and $x,y,z,u\in C_m$ where $x,y,z,u$ are not all the same.

Note that $0<a+b+c\leq 3\kp$ and $0< d \leq \kp$. From the first
coordinate of the equation above we have $a+b+c=d$.  The second
coordinate of the equation implies that
\[  \sum_{i=0}^{D-1} (ax_i+b y_i+cz_i) (\alpha K+1)^i \equiv  \sum_{i=0}^{D-1} d u_i (\alpha K+1)^i   \pmod{p}.\]
Since $0\leq ax_i+b y_i+ cz_i\leq 3 \kp K < \alpha K+1$ and $0\leq
du_i\leq \kp K <\alpha K+1$, necessarily $ax_i+by_i+cz_i=d u_i$ for
all $0\leq i \leq D-1$. Then
\begin{align*}
ax_i^2+by_i^2+cz_i^2 & = a (x_i-u_i+u_i)^2 + b (y_i-u_i+u_i)^2 +c (z_i-u_i+u_i)^2 \\
&=a (x_i-u_i)^2 +2a(x_i-u_i)u_i+au_i^2 +b (y_i-u_i)^2 +2b(y_i-u_i)u_i+bu_i^2  \\
&\  \ \ +c (z_i-u_i)^2 +2c(z_i-u_i)u_i+cu_i^2 \\
& = a (x_i-u_i)^2 +b (y_i-u_i)^2  +c (z_i-u_i)^2+2(ax_i+by_i+cz_i)u_i \\
& \ \ \  -(a+b+c) u_i^2 \\
& = a (x_i-u_i)^2 +b (y_i-u_i)^2  +c (z_i-u_i)^2 +(a+b+c) u_i^2.
\end{align*}  
Note that $x,y,z,u\in C_m$, i.e., $\sum_{i=0}^{D-1}x_i^2=\sum_{i=0}^{D-1}y_i^2=\sum_{i=0}^{D-1}z_i^2=\sum_{i=0}^{D-1}u_i^2$. It follows that
\begin{align*}
(a+b+c) \sum_{i=0}^{D-1}u_i^2 &=  a \sum_{i=0}^{D-1}x_i^2+b \sum_{i=0}^{D-1}y_i^2+c \sum_{i=0}^{D-1}z_i^2 \\
& = a  \sum_{i=0}^{D-1}(x_i-u_i)^2 +b \sum_{i=0}^{D-1} (y_i-u_i)^2  +c  \sum_{i=0}^{D-1}(z_i-u_i)^2 +(a+b+c)  \sum_{i=0}^{D-1}u_i^2,
\end{align*}
which in turn implies that $x_i=y_i=z_i=u_i$ for all $0\leq i\leq D-1$, and so, $x=y=z=u$, a contradiction.

\textbf{Case 6.}  $a s_x + b s_y  = c s_u +d s_v$ for some $a,b,c,d \in M$, $x,y,u,v\in C_m$, $x\neq y$ and $u \neq v$, where $abcd$ is positive. Note that from the first coordinate, we have $a+b=c+d$. By rearranging the terms, we may assume that 
\[a s_x + b s_y  = c s_u + d s_v\] for some $0<a,b,c,d \leq \kp$, $a+b=c+d$, $x,y,u,v\in C_m$ and $x,y,u,v$ are not all the same.
The second coordinate and the third coordinate of the equation above imply that
\begin{equation}\label{eq:1order}
ax+by\equiv cu+dv \pmod{p}
\end{equation}
and 
\begin{equation}\label{eq:2order}
ax^2+by^2\equiv cu^2+dv^2 \pmod{p}.
\end{equation}
We multiply \eqref{eq:2order} by $a+b$, and then subtract the square of \eqref{eq:1order}.  Noting that $a+b=c+d$, the result is 
\begin{equation}\label{eq:dif}
ab(x-y)^2\equiv cd (u-v)^2 \pmod{p}.
\end{equation}

If $x= y$, using \eqref{eq:1order} and \eqref{eq:dif}, it is easy to
see that $x,y,u,v$ are all the same, a contradiction; if $x=u$, then
\eqref{eq:1order} was ruled out in Case 4.  Thus, we may assume that
$x,y,u,v$ are pairwise distinct, and so,
\begin{equation}\label{eq:square}
abcd\equiv c^2d^2(u-v)^2/ (x-y)^2\pmod{p},
\end{equation}
i.e., $abcd$ should be a quadratic residue modulo $p$. 

Check all the possible $abcd$, where $0<a,b,c,d\leq \kp\leq 3$ and $a+b=c+d$. We have $abcd\in\set{1,2^2,3^2,4^2,6^2,9^2,3\cdot 2^2}$. Since $p\equiv \pm 5\pmod{12}$, $3$ is not a quadratic residue modulo $p$, and so, $abcd\in\set{1,2^2,3^2,4^2,6^2,9^2}$. In all of these cases, $abcd$ is a square in $\Z$.  Denote $t=\sqrt{abcd}$. Since $0<a,b,c,d\leq \kp$, we have $0<t\leq \kp^2$. Substituting $abcd=t^2$ in \eqref{eq:square} yields
\begin{equation}\label{eq:sqrt}
\pm t(x-y)\equiv cd(u-v) \pmod{p}.
\end{equation}

Solving the system of equations \eqref{eq:1order} and \eqref{eq:sqrt}, we get
\begin{equation*}
(c^2+cd)u\equiv (\pm t +ac)x+(bc\mp t)y \pmod{p}.
\end{equation*}
Note that $a+b=c+d$. Hence,
\begin{equation}\label{eq:sol} 
\sum_{i=0}^{D-1} (ac+bc)u_i (\alpha K+1)^i \equiv  \sum_{i=0}^{D-1}( (\pm t +ac)x_i+(bc\mp t)y_i )(\alpha K+1)^i \pmod{p}.
\end{equation}
Since $(\pm t +ac)+(bc\mp t)=ac+bc>0$, at least one of $\pm t +ac$ and $bc\mp t$ is positive. 
We proceed in the following  subcases.
\begin{enumerate}
\item If $\pm t +ac=0$, we have $t=ac$, and so,
\[\sum_{i=0}^{D-1} (ac+bc)u_i (\alpha K+1)^i \equiv \sum_{i=0}^{D-1}  (bc+ac)y_i (\alpha K+1)^i \pmod{p},\]
which in turn implies $u=y$, a contradiction.

Similarly, if $bc\mp t=0$, we can get $u=x$, again, a contradiction.
\item If both $\pm t +ac$ and $bc\mp t$ are positive, then $0\leq (\pm t +ac)x_i+(bc\mp t)y_i \leq (ac+bc) K \leq 2\kp^2 K\leq \alpha K$. On the other hand, $0\leq (ac+bc)u_i \leq 2\kp^2 K \leq \alpha K$. Thus it follows from \eqref{eq:sol} that 
\[(ac+bc)u_i= (\pm t +ac)x_i+(bc\mp t)y_i \textup{ for all } 0\leq i \leq D-1.\] That is,   the three distinct points $(x_0,x_1,\ldots, x_{D-1}), (y_0,y_1,\ldots, y_{D-1})$ and $(u_0,u_1,\ldots, u_{D-1})$ of $\F_p^D$ are collinear, which contradicts the fact that they are on the same sphere.
\item If $\pm t +ac$ is negative, then $bc\mp t$ is positive. Rearranging the terms in \eqref{eq:sol}, we have that 
\[\sum_{i=0}^{D-1} ((ac+bc)u_i-(\pm t +ac)x_i) (\alpha K+1)^i \equiv  \sum_{i=0}^{D-1}(bc\mp t)y_i (\alpha K+1)^i \pmod{p}.\] 
Since 
\begin{align*}
0 & \leq (ac+bc)u_i-(\pm t +ac)x_i = (ac+bc)u_i+(\mp t -ac)x_i  \\
 & \leq  (ac+bc)K+(\mp t -ac)K =  (bc \mp t) K \leq 2k_+^2 K \leq \alpha K,
\end{align*} 
then
\[(ac+bc)u_i- (\pm t +ac)x_i =(bc\mp t)y_i \textup{ for all } 0\leq i \leq D-1.\] 
Again we get three distinct points on the same sphere which are collinear, a contradiction.  
\item  If $bc\mp t$ is negative, then $\pm t +ac$ is positive. Using the same argument as above, we can get the contradiction.
\end{enumerate}
Thus we complete our proof.
\end{IEEEproof}

\begin{remark}
In the proof above, the product $abcd$ is required to be either a square of $\Z$ or a non quadratic residue modulo $p$. This requirement comes from Ruzsa's method, in the proof of Theorem of 7.3 of  \cite{Ruzsa93}. However, for $\kp\geq 4$, this requirement cannot be satisfied: we may choose  $(a,b,c,d)$ to be $(1,4,2,3), (1,3,2,2)$ or $(2,4,3,3)$, the products $24,12$ and $72$ are not squares  and they cannot simultaneously be non quadratic residues modulo $p$ for any prime $p$ as, using the Legendre symbol,
\[ \parenv*{\frac{6}{p}}=\parenv*{\frac{2}{p}}\parenv*{\frac{3}{p}}.\]
\end{remark}

\begin{corollary}
  \label{cor:behruz}
  Let $t=2$ and $0\leq \km \leq \kp \leq 3$, $\kp+\km\geq 1$.  There
  is an infinite family of $n$ such that $\Z^n$ can be lattice packed
  by $\BALL$ with density $\delta=\Omega(c^{-\sqrt{\ln n}})$ for some
  real number $c>0$.
\end{corollary}

\begin{IEEEproof}
Let $p\equiv\pm 5\pmod{12}$ be a sufficiently large prime. Set
$D\eqdef \floor{\sqrt{\ln p}}$ and $K\eqdef
\floor{(e^D-1)/\alpha}$. Then we have $(\alpha K+1)^D \leq p$.
Consider the group $G$ and the splitting sets $S_m$, $m\in
[0,DK^2-1]$, in Theorem~\ref{thm:con-t=2}. According to the pigeonhole
principle, there exists one set $S_m$ of size $\geq
\frac{(K+1)^D}{DK^2}$. Then
\begin{align*}
n & \geq \frac{(K+1)^D}{DK^2}  = \frac{\parenv*{\floor*{\frac{e^D-1}{\alpha}}+1}^D }{D\floor*{\frac{e^D-1}{\alpha}}^2} \geq   \frac{\parenv*{\frac{e^D-\alpha}{\alpha}+1}^D }{D\parenv*{\frac{e^D-1}{\alpha}}^2}
\geq
\frac{\alpha^2 e^{D^2}}{\alpha^{D}D e^{2D} } \\
& \geq \frac{\alpha^2 e^{  \floor*{\sqrt{\ln p}}^2}}{\alpha^{\sqrt{\ln p}} \sqrt{\ln p} e^{2\sqrt{\ln p}} }  \geq  \frac{\alpha^2 e^{  \parenv*{\sqrt{\ln p}-1}^2}}{\alpha^{\sqrt{\ln p}} \sqrt{\ln p} e^{2\sqrt{\ln p}} }  =   \frac{e\alpha^2 p}{\alpha^{\sqrt{\ln p}} \sqrt{\ln p} e^{4\sqrt{\ln p}} }\\
& \geq  \frac{p}{c_1 ^{ \sqrt{\ln p}} },
\end{align*}
for some real number $c_1>0$.  Taking logarithm of both side, we have
$\ln n \geq \ln p -\sqrt{\ln p} \ln c_1$, or equivalently,
\[\ln p - \ln c_1 \sqrt{\ln p} - \ln n \leq 0.\]
Solving for $\sqrt{\ln p}$, we get 
\[\sqrt{\ln p} \leq \frac{\ln c_1 +\sqrt{ (\ln c_1)^2+4\ln n } }{2},\]
and so
\[ \ln p \leq \frac{4\ln n +2(\ln c_1)^2 +2 \ln c_1 \sqrt{ (\ln c_1)^2+4\ln n } }{4}\leq  \ln n +c_2 \sqrt{\ln n},\]
for some real number $c_2>0$. It follows that 
\[ n\geq  \frac{p}{e^{c_2\sqrt{\ln n}}},\]
and 
\[\delta \geq  \frac{ \binom{n}{2} (\kp+\km)^2+n(\kp+\km)+1  }{\abs{G}} = \Omega(c^{-\sqrt{\ln n}}),\]
for some real number $c>0$.
\end{IEEEproof}

\section{Generalized Packings}
\label{sec:genpack}

It is a common practice in coding theory to also consider list
decoding instead of unique decoding. In such scenarios, the channel
output is decoded to produce a list of up to $\lambda$ possible
distinct codewords, where the channel output is within the error balls
centered at each of these codewords. In this section, we therefore
generalize the concept of packing to work in conjunction with list
decoding. The trade-off we present here is that at the price of a
constant-sized list, $\lambda$, we can find lattice packings of
$\BALL$ with density almost constant, $\Omega(n^{-\epsilon})$, for any
$\epsilon>0$. The proof method, however, is non-constructive, and
relies on the probabilistic method.

Given a shape $\cB \subseteq \Z^n$ and a lattice $\Lambda \subseteq
\Z^n$, we say $\cB$ \emph{$\lambda$-packs $\Z^n$ by $\Lambda$} if for
every element $\vz\in\Z^n$, there are at most $\lambda$ distinct
elements $\vv_i\in\Lambda$ such that $\vz\in\vv_i+\cB$. Obviously, if
$\lambda=1$, this definition coincides with the packing defined in
Section~\ref{sec:prelim}.
  
Let $G$ be a finite Abelian group, $M\eqdef [-\km,\kp]^*$, and
$S\subseteq G$.  If each element of $G$ can be written in at most
$\lambda$ ways as a linear combination of $t$ elements of $S$ with
coefficients from $M \cup \set{0}$, then we say $G
\overset{\lambda}{\geq} M \splt_t S$.
  
The following result is an analogue of Theorem~\ref{th:lattotile},
which relates lattice packings to Abelian groups. The proof is exactly
analogous, and we omit it.
  
\begin{theorem}
  Let $G$ be a finite Abelian group and $M\eqdef
  [-\km,\kp]^*$. Suppose that there is a subset
  $S=\set*{s_1,s_2,\ldots,s_n}\subseteq G$ such that $G
  \overset{\lambda}{\geq} M \splt_t S$.  Define $\phi:\Z^n\to G$ as
  $\phi(\vx)\eqdef\vx\cdot(s_1,\dots,s_n)$ and let
  $\Lambda\eqdef\ker\phi$ be a lattice.  Then $\cB(n,t,\kp,\km)$
  $\lambda$-lattice-packs $\Z^n$ by $\Lambda$.
\end{theorem}

We use the probabilistic approach detailed in~\cite{vu:2002}, and
follow some of the notation there. Let $x_1,x_2,\ldots, x_N$ be
independent $\set{0,1}$ random variables. Let
$Y=Y(x_1,x_2,\ldots,x_N)$ be a polynomial of $x_1,x_2,\ldots, x_N$.
$Y$ is \emph{normal} if its coefficients are between $0$ and $1$.  A
polynomial $Y$ is \emph{simplified} if every monomial is a product of
different variables. Since we are dealing with $\set{0,1}$ random
variables, every $Y$ has a unique simplification.  Given a set $A$,
let $\partial_A(Y)$ denote the partial derivative of $Y$ with respect
to $A$, and let $\partial_A^*$ be the polynomial obtained from the
partial derivative $\partial_A (Y)$ by subtracting its constant
coefficient. Define $\bE_j^*(Y)\eqdef \max_{\abs{A}\geq j}
\bE(\partial_A^* Y)$.

\begin{theorem}\cite[Corollary~4.9]{vu:2002}\label{thm:vucon}
  For any positive constants $\alpha$ and $\beta$ and a positive
  integer $d$, there is a positive constant $C=C(d,\alpha,\beta)$ such
  that if $Y$ is a simplified normal polynomial of degree at most $d$
  and $\bE_0^*(Y)\leq N^{-\alpha}$, then $\Pr(Y\geq C) \leq
  N^{-\beta}$.
\end{theorem}

We now use Theorem~\ref{thm:vucon} to show the existence of
generalized lattice packings with the desired parameters.

\begin{theorem}
  \label{th:prob}
  Let $0\leq \km\leq\kp$ with $\kp+\km\geq 1$, and $t>0$, be
  integers. Let $N$ be a sufficiently large integer such that
  $\gcd(N,\kp!)=1$, and fix $G\eqdef\Z_N$. Then for any
  $0<\epsilon<1/t$, there is a number $\lambda$ which only depends on
  $t$ and $\epsilon$, and a subset $S=\set{s_1,s_2,\dots,
    s_n}\subseteq G$ with $\frac{1}{2}N^{1/t-\epsilon}\leq n \leq
  \frac{3}{2}N^{1/t-\epsilon}$, such that $G \overset{\lambda}{\geq} M
  \splt_t S$, where $M\eqdef[-\km,\kp]^*$.
\end{theorem}

\begin{IEEEproof}
Set $\alpha=\epsilon t$, $\beta=2$, and $d=t$. Denote
\[p\eqdef N^{\frac{1}{t}-1-\epsilon}.\]
We construct $S$ randomly.  For each $0\leq i < N$, let the event
$i\in S$ be independent with probability $p$. Let $x_i$ be the
indicator variable of the event $i \in S$. Then
$\abs{S}=\sum_{i=0}^{N-1}x_i$, and
\[\bE(\abs{S})=Np=N^{\frac{1}{t}-\epsilon}.\]
Using Chernoff's inequality,  one can show that
\begin{equation}\label{eq:sizeexpc}
\Pr\parenv*{\frac{1}{2}\bE(\abs{S})\leq \abs{S}\leq \frac{3}{2}\bE(\abs{S})}\geq  1- 2e^{-\bE(\abs{S})/16}.
\end{equation}

For every $g\in G$ and $0 \leq i_1 < i_1 <\cdots <i_\ell < N$, denote
\[c(g;i_1,i_2,\ldots, i_\ell)\eqdef\abs{ \set*{(a_1,a_2,\ldots,a_\ell)\in M^\ell; g=a_1i_1+a_2i_2+\cdots a_\ell i_\ell} },\]
where addition and multiplication are in $G=\Z_N$. Consider the
following random variables (which are polynomials in the indicator random
variables $x_0,\dots,x_{N-1}$),
\[Y_g \eqdef    \parenv*{ \sum_{0 \leq i_1  <\cdots <i_t  < N } {c(g;i_1,i_2,\ldots,i_t)} x_{i_1}x_{i_2}\cdots x_{i_t}}/(\kp+\km)^t,\]
and
\[Z_g \eqdef   \parenv*{\sum_{\substack{1\leq \ell \leq t-1 \\ 0 \leq i_1  <\cdots <i_\ell < N }} {c(g;i_1,i_2,\ldots,i_{\ell})} x_{i_1}x_{i_2}\cdots x_{i_\ell}}/(\kp+\km)^{t-1}.\]
Both of them are positive, and as polynomials, they are simplified,
and normal. To show that $G \overset{\lambda}{\geq} M \splt_t S$, it
suffices to show that $(\kp+\km)^tY_g+(\kp+\km)^{t-1}Z_g \leq \lambda-1$
for every $g\in G$. 

We first look at $Y_g$. Since $\gcd(N,\kp!)=1$, if we fix
$a_1,a_2,\ldots,a_t\in M^t$ and $i_1,i_2,\ldots, i_{t-1}$, then there is a
unique $i_t \in [0,N-1]$ such that
$a_1i_1+a_2i_2+\cdots+a_ti_t=g$. Hence,
\[\bE(Y_g) \leq  ((\kp+\km)^t N^{t-1} p^{t})/(\kp+\km)^t= N^{-\epsilon t} =N^{-\alpha}.\]
For the partial derivative  $\partial_A (Y_g)$ with $A=\set{j_1,j_2,\ldots, j_k} \subseteq [0,N-1]$ and $k\leq t-1$,
\[ \partial_A (Y_g)= \sum_{c_1,c_2,\ldots, c_k \in M}\parenv*{ \sum_{0\leq i_1<i_2 <\cdots i_{t-k}<N}   {c(g-c_1j_1-\cdots-c_kj_k;i_1,i_2,\ldots,i_{t-k})} x_{i_1}x_{i_2}\cdots x_{i_{t-k}}}/(\kp+\km)^t, \]
hence,
\begin{align*}
\bE(\partial_A (Y_g)) & \leq (\kp+\km)^k((\kp+\km)^{t-k}N^{t-k-1}p^{t-k}/(\kp+\km)^t) \\
& = N^{-\epsilon t +k\epsilon -k/t} < N^{-\alpha}.
\end{align*}
Applying  Theorem~\ref{thm:vucon}, there is a number $C$, depending on $t$ and $\epsilon$, such that 
\begin{equation}\label{eq:Yupbound}
\Pr(Y_g \geq C) \leq  N^{-2}.
\end{equation}

As for $Z_g$, a similar computation to the above shows that 
\[\bE(\partial_A^*(Z_g))=  O(N^{t-1-k-1} p^{t-1-k}).\]
Since 
\begin{align*}
N^{t-1-k-1} p^{t-1-k}=  N^{-\epsilon t-(\frac{1}{t}-\epsilon)(k+1)} < N^{-\alpha},
\end{align*}
we have $\bE(\partial_A^*(Z_g)) < N^{-\alpha}$ when $N$ is sufficiently large. Applying Theorem~\ref{thm:vucon} again, there is a $C'$ such that 
\begin{equation}\label{eq:Zupbound}
\Pr(Z_g\geq C') \leq N^{-2}.
\end{equation}
Denote $\lambda = C(\kp+\km)^t +C'(\kp+\km)^{t-1}+1$. Then  \eqref{eq:Yupbound}  and \eqref{eq:Zupbound} imply, via a union bound, that
\[\Pr(\textup{there exists } g \in G  \textup{ such that }  (\kp+\km)^tY_g+(\kp+\km)^{t-1}Z_g  > \lambda-1) \leq 2N^{-1}. \]
From this, together with \eqref{eq:sizeexpc}, we can see that the
random set $S$ satisfies the conditions with probability at least
$1-2N^{-1}-2e^{-\bE(\abs{S})/16}$, which is positive for large enough
$N$. Thus, such a set exists.
\end{IEEEproof}

\begin{corollary}
  \label{cor:genpack}
  Let $0\leq \km\leq\kp$ with $\kp+\km\geq 1$, and $t>0$, be
  integers. Then for any real number $\epsilon>0$, there is an integer
  $\lambda$ and infinitely many values of $n$ such that $\Z^n$ can be
  $\lambda$-lattice-packed by $\BALL$ with density
  $\delta=\Omega(n^{-\epsilon})$.
\end{corollary}
\begin{IEEEproof}
  Fix $\epsilon'\eqdef\frac{\epsilon}{\epsilon t+ t^2}$, and observe
  that $\epsilon'<1/t$. Use Theorem~\ref{th:prob} with $\epsilon'$,
  noting that $N=\Theta(n^{\epsilon+t})$, to obtain a lattice
  $\lambda$-packing of $\BALL$ with density
  \[ \delta \geq \frac{\abs{\BALL}}{\abs{G}}=\frac{\sum_{i=0}^t \binom{n}{i}(\kp+\km)^i}{N}=\Omega(n^{-\epsilon}).\]
\end{IEEEproof}

As a final comment on the matter, we observe that a tedious
calculation shows that in the above corollary
$\lambda=O(\epsilon^{-t})$ -- a calculation which we omit.


\section{Constructions of Lattice Coverings}
\label{sec:covering}

We switch gears in this section, and focus on covering instead of
packing. We first argue that using known techniques from the theory of
covering codes in the Hamming metric, we can show the existence of
non-lattice coverings of $\Z^n$ by $\BALL$. However, these have a high
density of $\Omega(n)$. We then provide a product construction to
obtain a lattice covering by $\BALL$ with density $O(1)$.

Fixing an integer $\ell\in\N$, we use the same argument as the one given
in \cite[Section~12.1]{CohHonLitLob97} to construct a covering
code $C\subseteq\Z_\ell^n$, of size
\[ \abs*{C}=\ceil*{\frac{n\ell^n \ln \ell}{\abs*{\BALL}}}.\]
We can then translate this covering of $\Z_\ell^n$ by $\BALL$ to a covering
of $\Z^n$ by using the same idea as Theorem~\ref{th:linearcode}, and
defining $C'\eqdef\set*{ \vx\in\Z^n ; (\vx\bmod \ell)\in C}$. However,
the density of the resulting covering is
\[ \delta= \frac{\abs*{C}\cdot\abs*{\BALL}}{\ell^n} = \Omega(n).\]
We therefore proceed to consider more efficient coverings using the
product construction whose details follow.

\begin{theorem}
  \label{con-reccovering}
  Suppose that there exist a finite Abelian group $G$ and a subset $S
  \subseteq G$ such that $G \leq M \splt_1 S$. Let $t>0$ be an
  integer, and denote
  \[S^{(t)}\triangleq\set*{(s,0,0,\dots, 0) ; s \in S} \cup \set*{(0,s,0,\dots,0) ; s\in S} \cup \cdots \cup\set*{(0,0,0,\dots,s) ; s \in S}.\]
  Then $G^{t}\leq  M \splt_t S^{(t)}$.
\end{theorem}

\begin{IEEEproof}
  For any element $g=(g_1,g_2,\ldots,g_t)\in G^{t}$, since $G\leq M
  \splt_1 S$, for each $1\leq i \leq t$, there are $s_i\in S$ and $c_i
  \in M\cup \set{0}$ such that $g_i=c_i\cdot s_i$.  Hence,
  \[g=c_1(s_1,0,0,\ldots,0)+c_2(0,s_2,0,\ldots,0)+\cdots+c_t(0,0,0,\ldots,s_t).\]
  That is, $g$ can be written as a linear combination of $t$ elements
  of $S$ with coefficients from $M \cup \set{0}$, and so, $G^{t}\leq M
  \splt_t S^{(t)}$.
\end{IEEEproof}

We can now construct a lattice covering, using the previous theorem.

\begin{corollary}
  \label{cor:cover}
Let $\opsi(x)$ be the largest prime not larger than $x$, and denote
$p\eqdef\opsi(\kp+\km+1)$. Let $0\leq \km\leq \kp$, with $\kp+\km\geq
1$, and $t>0$, be integers. Define $M\eqdef [-\km,\kp]^*$. Then for
any integer $m>0$, there exists $S\subseteq (\Z_{p^m})^t$,
$\abs{S}=t\cdot\frac{p^m-1}{p-1}$, such that $(\Z_{p^m})^t\leq M
\splt_t S$, and thereby, a lattice covering of $\Z^n$, $n=\abs{S}$,
with density
\[\delta = \frac{ \sum_{i=0}^{t} \binom{n}{i} (\kp+\km)^i }{  ( n(p-1)/t+1  )^t  }=\frac{(t(\kp+\km))^t}{t!(p-1)^t}+o(1).\]
\end{corollary}

\begin{IEEEproof}
 According to \cite[Construction 1]{Sch12}, there is a subset
 $A\subseteq \Z_{p^m}$ of size $\frac{p^m-1}{p-1}$ such that
 $\Z_{p^m}=[-\km,p-1-\km]^*\splt_1 A$. We may apply
 Theorem~\ref{con-reccovering} to obtain a subset $S=A^{(t)}\subseteq
 (\Z_{p^m})^t$ of size $t\cdot\frac{p^m-1}{p-1}$ such that
 $(\Z_{p^m})^t\leq M\splt_t S$. The calculation of the density of the
 resulting lattice covering is straightforward.
\end{IEEEproof}



\section{Conclusion}
\label{sec:conclusion}

Motivated by coding for integer vectors with limited-magnitude errors,
we provided several constructions of packings of $\Z^n$ by $\BALL$ for
various parameters. These are summarized in
Table~\ref{tab:summary}. While the parameter ranges of the
constructions sometimes overlap, and perhaps result in equal or
inferior asymptotic density, having more constructions allows for more
choices for fixed values of the parameters.

One main goal was to construct lattice packings, analogous to linear
codes, as these are generally easier to analyze, encode, and
decode. Thus, except for one case, all constructions we provide
lattices. The main tool in constructing these is the connection
between lattice packings of $\BALL$ and $t$-splittings of Abelian
groups. The other important goal was to have asymptotic packing
density that is non-vanishing. This is achieved in many of the cases.

We also discussed $\lambda$-packing, which allows for a small overlap
between the translates of $\BALL$ centered at the lattice points. This
is useful for list-decoding setting with a list size of $\lambda$. The
result we obtain is non-constructive, and it provides a trade-off
between the list size and the packing density. Finally, we also
addressed the problem of lattice-covering of $\Z^n$ by $\BALL$,
showing using the product construction, that there exist such
coverings with asymptotic constant density.

The results still leave numerous open questions, of which we mention
but a few:
\begin{enumerate}
\item
  Constructions for packings of $\BALL$ with $t\geq 3$, $\kp\geq 2$,
  and non-vanishing asymptotic density are still unknown.
\item
  Whether asymptotic density of $1$ is attainable for all parameters
  is still an open questions. Such lattice packings would be analogous
  to asymptotically perfect codes.
\item
  In the asymptotic regime of $t=\Theta(n)$, all of the constructions
  in this paper produce packings with vanishing asymptotic rates. Such
  families of packings are analogous to good codes in the Hamming
  metric, and their existence and constructions would be most welcome.
\item
  Efficient decoding algorithms are missing for most of the cases. In
  the asymptotic regime of constant $t$,
  $\abs{\BALL}=\Theta(n^t)$. Thus, we are looking for non-trivial
  decoding algorithms, whose run-time is $o(n^t)$.
\item
  We would also like to find constructive versions of the non-constructive
  proofs for $\lambda$-packings, and covering lattices.
\end{enumerate}

\begin{table*}
  \caption{A summary of the results}
  \label{tab:summary}
  {\renewcommand{\arraystretch}{1.2}
  \begin{tabular}{cccccll}
    \hline\hline
    Type    & $t$ & $\kp$ & $\km$ & density & Location & Comment \\
    \hline
    Packing & any & 1 & 0 & $\frac{1}{t!}+o(1)$ & Corollary~\ref{cor:lc} & via BCH codes \\
    Packing & any & any & any & $\Theta(n^{t-\ceil{2t(1-{1}/p)}})$ & Corollary~\ref{cor:lc} & via BCH codes, $p\geq \kp+\km+1$ prime\\
    Packing & any & 1 & 0 & $\frac{1}{t!}+o(1)$ & Corollary~\ref{cor:btsets} & via $B_t[N;1]$ sets \\
    Packing & any & 1 & 1 & $\frac{2^t}{t!(2t+1)}+o(1)$ & Corollary~\ref{cor:kpm1} & via $B_t[N;1]$ sets \\
    Packing & 2 & 1 & 0 & $1-o(1)$ & Corollary~\ref{cor:perparata} & via Preparata codes, non-lattice packing \\
    Packing & 2 & 1  & 1 & $\frac{1}{2}+o(1)$ & Corollary~\ref{cor:quasiperfectlc} & via quasi-perfect linear codes\\
    Packing & 2 & 2 & 0 & $\frac{1}{2}+o(1)$ & Corollary~\ref{cor:quasiperfectlc} & via quasi-perfect linear codes\\
    Packing & 2 & $\leq 2$ & $\leq 2$ & $\frac{1}{8(2(\kp+\km)+1)}+o(1)$ & Corollary~\ref{cor:t2sidon} & via $2$-fold Sidon sets \\
    Packing & 2 & $\leq 3$ & $\leq 3$ & $\Omega(c^{-\sqrt{\ln n}})$ & Corollary~\ref{cor:behruz} & via Behrend's and Ruzsa's methods \\
    $\lambda$-Packing & any & any & any & $\Omega(n^{-\epsilon})$ & Corollary~\ref{cor:genpack} & $\lambda=O(\epsilon^{-t})$ \\
    Covering & any & any & any & $\frac{(t(\kp+\km))^t}{t!(p-1)^t}+o(1)$ & Corollary~\ref{cor:cover} & $p\leq \kp+\km+1$ prime \\
    \hline\hline
  \end{tabular}
  }
\end{table*}

\end{document}